\documentclass{amsart}

\usepackage{amssymb}

\usepackage{amsthm}

\usepackage{amsfonts}
\usepackage{amsmath}
\def\CC{{\mathbb{C}}}
\def\QQ{{\mathbb{Q}}}
\def\RR{{\mathbb{R}}}
\def\ZZ{{\mathbb{Z}}}
\def\NN{{\mathbb{N}}}

\def\KK{{\mathbb{K}}}

\def\bigO{\mathcal{O}}
\def\bigOsoft{\tilde{\mathcal{O}}}

\newtheorem{Thm}{Theorem}
\newtheorem{Lem}[Thm]{Lemma}
\newtheorem{Prop}[Thm]{Proposition}

\theoremstyle{definition}
\newtheorem{Def}[Thm]{Definition}

\theoremstyle{remark}
\newtheorem{Rem}[Thm]{Remark}

\theoremstyle{remark}

\begin{document}

\title[Recombination for Decomposition]{A Recombination algorithm  for the decomposition of multivariate rational functions}

\date{\today}

\author[G.~Ch\`eze]{Guillaume Ch\`eze}
\address{Institut de Math\'ematiques de Toulouse\\
Universit\'e Paul Sabatier Toulouse 3 \\
MIP B\^at 1R3\\
31 062 TOULOUSE cedex 9, France}
\email{guillaume.cheze@math.univ-toulouse.fr}

\begin{abstract}
In this paper we show how we can compute in a deterministic way the decomposition of a multivariate rational function  with a recombination strategy.  The key point of our recombination strategy is the used of Darboux polynomials. We study the complexity of this strategy and we show that this method improves the previous ones.  In appendix, we explain how the strategy proposed recently by J.~Berthomieu and G.~Lecerf for the sparse factorization can be used in the decomposition setting. Then we deduce a decomposition algorithm in the sparse bivariate case and we give its complexity.
\end{abstract}

\maketitle

%%%%%%%%%%%%%%%%%%%%%%%%%%%%%%%%%%%%%%%%

\section*{Introduction}

The decomposition of univariate polynomials has been widely studied since 1922, see \cite{Ritt}, and efficient algorithms are known, see \cite{Alagar_Thanh, Barton_Zippel, Kozen_Landau, Gathen_tame, Gathen_wild,Giesbrecht,Kluners}. There also exist results and algorithms in the multivariate case \cite{Dickerson, Gathen_tame,Giesbrecht,GatGutRub}. \\The decomposition of rational functions has also been studied,  \cite{Giesbrecht,Zippel, Alonso_Gut_Recio,Gathen_hom_biv}. In the multivariate case the situation is the following:\\
Let $f(X_1,\dots, X_n)=f_1(X_1,\dots,X_n)/f_2(X_1,\dots,X_n) \in \KK(X_1,\dots,X_n)$ be a rational function, where $\KK$ is a field and $n\geq 2$. It is commonly said to be composite if it can be written $f=u(h)$ where $h(X_1,\dots,X_n) \in \KK(X_1,\dots,X_n)$ and $u \in\KK(T)$ such that $\deg(u) \geq 2$ (recall that the degree of a rational function is the maximum of the degrees of its numerator and denominator after reduction), otherwise $f$ is said to be non-composite. In this paper, we give an algorithm which computes a non-composite rational function $h \in \KK(X_1,\ldots,X_n)$ and a rational function $u \in \KK(T)$ such that $f=u(h)$.

\medskip
In \cite{Cheze2010}, the author shows that we can reduce the decomposition problem to a factorization problem and gives  a probabilistic  and a deterministic algorithm. The probabilistic algorithm is nearly optimal: it performs $\bigOsoft(d^n)$ arithmetic operations. The deterministic one computes $\bigO(d^2)$ absolute factorizations and then performs $\bigOsoft(d^{n+\omega+2})$ arithmetic operations, where $d$ is the degree of $f$ and $ \omega$ is  the \emph{feasible matrix
multiplication} exponent as defined in~\cite[Chapter~12]{GG}. We recall that $ 2 \leq \omega \leq 2.376$. As in  \cite{Cheze2010}, we suppose in this work that $d$ tends to infinity and $n$ is fixed.  We use the classical $\bigO$ and $\bigOsoft$ (``soft $\bigO$'') notation in the neighborhood of infinity as defined in~\cite[Chapter~25.7]{GG}. Informally speaking, ``soft $\bigO$''s are used for readability in order to hide logarithmic factors in complexity estimates. In this paper we improve the complexity of the deterministic algorithm. With this algorithm we just compute two factorizations in $\KK[X_1,\ldots,X_n]$ and then we use a recombination strategy. Under some hypotheses this new method performs $\bigOsoft(d^{n+\omega-1})$ arithmetic operations.  

\medskip
The decomposition of multivariate rational functions appears when we study the kernel of a derivation, see \cite{Ollagnier}.  In \cite{Ollagnier} the author uses Darboux polynomials and gives an algorithm which  works with $\bigOsoft(d^{\omega n})$ arithmetic operations.\\
 In this paper, we are also going to use  Darboux polynomials (see Section \ref{toolbox} for a definition) and we add a recombination strategy. Roughly speaking, we are going to factorize the numerator and the denominator and then thanks to  a property of Darboux polynomials we are going to show that we can recombine the factors and deduce the decomposition.\\

The decomposition of multivariate rational functions also appears when we study intermediate fields of an unirational field and the extended L\"uroth's Theorem, see \cite{Gut_Rub_Sev, Cheze2010}, and when  we study the spectrum of a rational function, see  \cite{Cheze2010} and the references therein.

The study of decomposition is an active area of research: for a study on multivariate polynomial systems see e.g. \cite{faugere_perret,faugere_vonzur}, for a study on symbolic polynomials see e.g. \cite{watt}, for a study on Laurent polynomials see e.g. \cite{watt_laurent}, for effective results on the reduction modulo a prime number of a non-composite polynomial or a rational function see e.g. \cite{ChezeNajib,BodinDebesNajib,BuseChezeNajib}, for combinatorial results see e.g. \cite{vonzur_denom}.

\medskip

In this paper, we improve the strategy proposed in  \cite{Ollagnier}. As in \cite{Ollagnier}, we consider fields with characteristic zero.  Furthermore, as we want to give precise complexity estimate we are going to suppose that:\\

\begin{center}{ Hypothesis (C):\\
  $\KK$ is a number field: $\KK=\QQ[\alpha]$, $\alpha$ is an algebraic number of degree $r$. }\\
\end{center}

\medskip

As in \cite{Cheze2010}, we are  going to suppose that the following hypothesis is satisfied:\\

\begin{center}{Hypothesis  (H):}\end{center}
$$\left\{
\begin{array}{l}
(i)  \deg (f_1+\Lambda f_2)=\deg_{X_n}(f_1 + \Lambda f_2), \textrm{ where } \Lambda \textrm{ is a new variable},\\
(ii)  \, R(\Lambda)=Res_{X_n} \Big( f_1(\underline{0},X_n)+\Lambda f_2(\underline{0},X_n), \, \partial_{X_n} f_1(\underline{0},X_n) + \Lambda  \partial_{X_n} f_2(\underline{0},X_n)\Big) \neq 0 \textrm{ in } \KK[\Lambda].
\end{array}
\right.
$$
%\begin{center}
%Hypothesis (H): $\deg_{X_n} g < \deg_{X_n}f = \deg f$ and $ f$  is monic in $ X_n$,
%\end{center}
where $\deg_{X_n} f$ represents the partial degree of $f$ in the variable $X_n$, $\deg f$ is the total degree of $f$ and $Res_{X_n}$ denotes the resultant relatively to the variable $X_n$.

\medskip

This hypothesis is necessary, because we will use the factorization algorithms proposed in \cite{Lec2007}, where this kind of hypothesis is needed.  Actually, in \cite{Lec2007} the author studies the factorization of a polynomial $F$ and uses hypothesis (L), where (L) is the following:
\begin{center}{ Hypothesis  (L):}\end{center}
$$\left\{
\begin{array}{l}
(i) \deg_{X_n}F=\deg F, \, \textrm{ and } F \textrm{ is monic in } X_n, \\
(ii)  Res_{X_n} \big( F(\underline{0},X_n), \frac{\partial F}{\partial X_n}(\underline{0},X_n) \Big) \neq 0.
\end{array}
\right.
$$
If $F$ is squarefree, then  hypothesis (L)  is not restrictive since it  can be assured by means of a generic linear change of variables, but we will not discuss this question here (for a complete treatment in the bivariate case, see \cite[Proposition 1]{CL}).

\medskip

Roughly speaking, our hypothesis (H) is the hypothesis (L) applied to the polynomial $f_1 + \Lambda f_2$. %Furthermore, if $f_1$ is such that: $\deg (f_1)=\deg(f_1/f_2)$ and $f_1$ satisfies (L,$ii$) then (H,$ii$) is satisfied.\\
 In (H,$i$) we do not assume that $f_1+ \Lambda f_2$ is monic in $X_n$. Indeed,  the leading coefficient relatively to $X_n$ can be written: $a+\Lambda b$, with $a,b \in \KK$. In our  algorithm, we evaluate $\Lambda$ to $\lambda$ such that $a + \lambda b \neq 0$. Then we can consider the monic part of $f_1+ \lambda f_2$ and we get a polynomial satisfying (L,$i$). Then (H,$i$) is sufficient in our situation. Furthermore, in this paper, we assume $f_1/f_2$ to be  reduced, i.e. $f_1$ and $f_2$ are coprime. We recall in Lemma \ref{Cor_f+gsquarefree} that in this situation $f_1 + \Lambda f_2$ is squarefree. Thus hypothesis (H) is not restrictive.

Furthermore, hypothesis (H) will also be useful in a preprocessing step, see Section \ref{sec_reduc}. In this preprocessing step we reduce the decomposition to two factorizations of squarefree polynomials.

\subsection*{Complexity model}
In this paper the  complexity estimates charge a constant
cost for each arithmetic operation ($+$, $-$, $\times$, $\div$) and
the equality test. All the constants in the base fields   are thought to be freely at our disposal.
\par
%We use the classical $\bigO$ and $\bigOsoft$ (``soft $\bigO$'') notation in
%the neighborhood of infinity as defined
%in~\cite[Chapter~25.7]{GG}.  Informally speaking, ``soft $\bigO$''s
%are used for readability in order to hide logarithmic factors in
%complexity estimates. 
In this paper we suppose that \emph{the number of variables $n$ is fixed} and that the degree $d$ tends to infinity. 
\par
Polynomials are represented by dense vectors of their
coefficients in the usual monomial basis. For each integer $d$, we
assume that we are given a computation tree that computes the product
of two univariate polynomials of degree at most $d$ with at most $\bigOsoft(d)$
 operations, independently of the base ring, see \cite[Theorem~8.23]{GG}. Then with a Kronecker substitution we can compute the product of two multivariate polynomials with degree $d$ with $n$ variables with $\bigOsoft(d^n)$ arithmetic operations. We also recall, see \cite[Corollary 11.8]{GG}, that if $\KK$ is an algebraic extension of $\QQ$ of degree $r$ then each field operation in $\KK$ takes $\bigOsoft(r)$ arithmetic operations in $\QQ$.\\
%In \cite{CheLec} the autors give a probabilistic algorithm for the absolute factorization of a bivariate polynomial with $\bigOsoft(d^3)$ operations in $\KK$.\\
We use the constant $\omega$ to denote a \emph{feasible matrix
multiplication} exponent as defined in~\cite[Chapter~12]{GG}:
two~$n\times n$ matrices over $\KK$ can be multiplied
with~$\bigO(n^\omega)$ field operations. As in~\cite{BiPa1994}
we require that $2 \le \omega \le 2.376$. We recall that the computation of a solution basis  of a linear system with $m$ equations and $d\leq m$ unknowns over $\KK$ takes $\bigO(md^{\omega-1})$ operations in $\KK$ \cite[Chapter~2]{BiPa1994} (see
also~\cite[Theorem~2.10]{Storjohann2000}).\\
In \cite{Lec2006,Lec2007} the author gives a deterministic  algorithm for the multivariate rational factorization. The rational factorization of a polynomial $f$ is the factorization in $\KK[\underline{X}]$, where $\KK$ is  the coefficient field of $f$. This algorithm uses one factorization of a univariate polynomial of degree $d$ and $\bigOsoft(d^{n+\omega-1})$ arithmetic operations, where $d$ is the total degree of the polynomial and $n \geq 2$ is the number of variables.

\subsection*{Main Theorem}
The following theorem gives the complexity result of our algorithm. 
\begin{Thm}\label{main_thm}
Let $f$ be a multivariate rational function in $\QQ[\alpha](X_1, \dots,X_n)$ of degree $d$, where $\alpha$ is an algebraic number of degree $r$.
Under hypotheses (C) and (H), we can compute in a deterministic way the decomposition of $f$ with $\bigOsoft(rd^{n+ \omega -1})$ arithmetic operations over $\QQ$ plus two factorizations of  univariate polynomials of degree $d$ with coefficients in $\QQ[\alpha]$.
\end{Thm}

%%%%%%%%%%%%%%%%%%%%%%%%%%%%%%%%%%%%%%%%%%%%
\subsection*{Comparison with other algorithms}
There already exist several algorithms for the decomposition of rational functions.  They all use the same global strategy: first compute $h$, and then deduce $u$. The first step is the difficult part of the problem. In \cite{Cheze2010}, we explain how we can perform the second step, i.e. compute $u$ from $h$ and $f$, with $\bigOsoft(d^n)$ arithmetics operations.\\
In \cite{Gut_Rub_Sev}, the authors provide two algorithms to decompose a multivariate rational function. These algorithms run in exponential time in the worst case. In the first one we have to factorize polynomials with $2n$ variables \mbox{$f_1(\underline{X})f_2(\underline{Y})-f_1(\underline{Y})f_2(\underline{X})$} and to look for factors of the following kind $h_1(\underline{X})h_2(\underline{Y})-h_1(\underline{Y})h_2(\underline{X})$. The authors say that in the worst case the number of candidates to be tested is exponential in $d=\deg( f_1/f_2)$. Indeed, the authors test all the possible factors.\\
  In the second algorithm, for each pair of factors $(h_1,h_2)$ of $f_1$ and $f_2$ (i.e. $h_1$ divides $f_1$ and $h_2$ divides $f_2$), we have to test if there exists $u \in \KK(T)$ such that $f_1/f_2=u(h_1/h_2)$. Thus in the worst case we also have an exponential number of candidates to be tested. \\
To the author's knowledge, the first polynomial time algorithm is due to J.~Moulin-Ollagnier, see \cite{Ollagnier} . This algorithm relies on the study of the kernel of the following derivation: $\delta_{\omega}(F)=\omega \wedge dF$, where $F \in \KK[\underline{X}]$ and $\omega=f_2df_1-f_1df_2$. In \cite{Ollagnier} the author shows that we can reduce the  decomposition of a rational function to linear algebra.  The bottleneck of this algorithm is the computation of the kernel of a matrix. The size of this matrix is $\bigO(d^n) \times \bigO(d^n)$, then the complexity of this deterministic algorithm belongs to $\bigO(d^{n\omega})$. \\

The reduction of the decomposition problem to a factorization problem is classical, see e.g. \cite{Kluners,Giesbrecht,Gathen_hom_biv, Gut_Rub_Sev}. In  \cite{Cheze2010} the author shows that  if we choose a probabilistic approach  then two factorizations in $\KK[X_1,\ldots,X_n]$ are sufficient to get $h$ and furthermore we do not have a recombination problem. This gives a nearly optimal algorithm. For the deterministic approach the author uses a property on the pencil $f_1 - \lambda f_2$ and shows that with $\bigO(d^2)$ absolute factorization (i.e. factorization in the algebraic closure of $\KK$) we can get $h$. This deterministic strategy works with $\bigOsoft(d^{n+\omega +2})$ arithmetic operations.\\
In this paper, we are going to show that we can obtain a deterministic algorithm with just two factorizations in $\KK[X_1,\ldots,X_n]$ and a recombination strategy. Our algorithm uses at most $\bigOsoft(d^{n+\omega-1})$ arithmetic operations. This cost corresponds to the cost of the factorization and the recombination step.\\

Our recombination problem comes from this factorization:\\
If $f_1/f_2=u_1/u_2(h_1/h_2)$ then 
$$f_1 - \lambda f_2= e (h_1 - t_1 h_2)\cdots (h_1 -t_k h_2)$$
where $\lambda, e \in \KK$, $k= \deg (u_1/u_2)$ and $t_i$ are the roots of the univariate polynomial $u_1(T) - \lambda u_2(T)$, see Lemma \ref{Lem_Bodin}.\\
Thus with the factors $h_1 -t_1 h_2$ and $h_1 -t_2 h_2$ we can deduce $h$. Unfortunately these factors are not necessarily in $\KK[X_1,\ldots,X_n]$ and are not necessarily irreducible. In this paper we show how we can reduce the problem to a factorization problem in $\KK[X_1,\ldots,X_n]$ and how we can recombine the irreducible factors of $f_1 -\lambda f_2$ to get $h$.\\

We can see our recombination scheme as a \emph{logarithmic derivative method}. 
Roughly speaking, the logarithmic derivative method works as follow:\\
If $F(X,Y)=\prod_{j=1}^t \mathcal{F}_j(X,Y)$, where $\mathcal{F}_j(X,Y) \in \mathbb{A}$ and $\mathbb{A} \supset \KK[X,Y]$ (for example $\mathbb{A}=\KK[[X]][Y]$), then we can write the irreducible factors $F_i(X,Y) \in \KK[X,Y]$ of $F$ in the following way: $F_i=\prod_{j=1}^t \mathcal{F}_j^{e_{i,j}}$, where $e_{i,j} \in \{0,1\}$. Thus we just have to compute the exponents $e_{i,j}$ to deduce $F_i$. We compute these exponents thanks to this relation:
$$\dfrac{\partial_X F_i}{F_i}=\sum_{j=1}^t e_{i,j} \dfrac{\partial_X \mathcal{F}_j}{\mathcal{F}_j}.$$
With this relation the exponents $e_{i,j}$ are now coefficients, and we can compute them with linear algebra.\\
This strategy has already been used by several authors in order to factorize polynomials see e.g. \cite{Bellabas,Bostan,Lec2006,CL,Weiman}. Here, we use this kind of technique for the decomposition problem. With this strategy the recombination part of our algorithm corresponds to the computation of the kernel of a $\bigO(d^n) \times \bigO(d)$ matrix. \\
 In our context, we do not use exactly a logarithmic derivative. We use a more general derivation, but we  use the same idea: if a mathematical object transforms a product into a sum then the recombination problem becomes a linear algebra problem. In this paper this mathematical object is the cofactor, see Proposition \ref{prop_cof}.

\subsection*{Structure of this paper} In Section 1, we recall some results about the Jacobian derivative and Darboux polynomials. In Section 2, we describe a reduction step which eases the recombination strategy. In other words we explain how we can  reduce the decomposition problem to a factorization problem. In Section 3, we explain how we can get $h$ with a recombination strategy. In Section \ref{sec:examples}, we describe our algorithm with two examples. In Section \ref{sec:conclusion} we conclude this paper with a remark on Darboux method and the logarithmic derivative method.  In appendix, we explain how the strategy proposed recently by J.~Berthomieu and G.~Lecerf  in \cite{BL2010} for the sparse factorization can be used in the decomposition setting. Then we deduce a decomposition algorithm in the sparse bivariate case and we give its complexity.

\subsection*{Notations}
All the rational functions are supposed to be reduced.\\
Given a polynomial $f$, $\deg(f)$ denotes its total degree.\\
Given a rational function $f=f_1/f_2$, $\deg(f)$ denotes $\max\big(\deg(f_1), \deg(f_2)\big)$.\\ 
%$\overline{\KK}$ is an algebraic closure of $\KK$.\\
For the sake of simplicity, sometimes we write $\KK[\underline{X}]$ instead of $\KK[X_1,\dots,X_n]$, for $n\geq 2$.\\
$u \circ h$ means $u(h)$.\\
$Res(A,B)$ denotes the resultant of two univariate polynomials $A$ and $B$.\\
$|S|$ is the cardinal of the set $S$.

%%%%%%%%%%%%%%%%%%%%%%%%%%%%%%%%%%%%%%%%%%%%%%%%%%
%%%%%%%%%%%%%%%%%%%%%%%%%%%%%%%%%%%%%%%
\section{Derivation and Darboux polynomials}\label{toolbox}
%%%%%%%%%%%%%%%%%%%%%%%%%%%%%%%%%%%%%%%%%%
We introduce the main tool of our algorithm. 

\begin{Def}
A $\KK$-derivation $D$ of the polynomial ring $\KK[X_1,\ldots,X_n]$ is a $\KK$-linear map from $\KK[X_1,\ldots,X_n]$ to itself that satisfies the Leibniz rule for the product
$$D(f.g)=D(f).g+f.D(g).$$
\end{Def}
A $\KK$-derivation has a unique  extension to $\KK(X_1,\ldots,X_n)$ and then we will also denote by $D$ the extended derivation. 
\begin{Def}
Given a rational function $f_1/f_2$, the Jacobian derivative associated to $f_1/f_2$ is the following vector derivation, i.e. an $(n-1)$-tuple of derivations:
\begin{eqnarray*}
D_{f_1/f_2}: \KK[X_1,\ldots,X_n]& \longrightarrow& \Big(\KK[X_1,\ldots,X_n]\Big)^{n-1}\\
F&\longmapsto& f_2^2.\begin{pmatrix}\partial_{X_1}\big(f_1/f_2)\partial_{X_2} F- \partial_{X_2}(f_1/f_2)\partial_{X_1}F\\
\vdots\\
\partial_{X_1}\big(f_1/f_2)\partial_{X_n} F- \partial_{X_n}(f_1/f_2)\partial_{X_1}F
\end{pmatrix}.
\end{eqnarray*}
\end{Def}

The Jacobian derivative has the following property:

\begin{Prop}\label{Petrav}
Given $f=f_1/f_2$ and $g \in \KK(X_1,\ldots,X_n)\setminus \KK$ the following propositions are equivalent:
\begin{enumerate}
%\item $f$ and $g$ are algebraically dependent over $\KK$;
\item The rank of the Jacobian matrix 
$$Jac(f,g)=\begin{pmatrix} 
\dfrac{\partial f}{\partial X_1} &\cdots &\dfrac{\partial f}{\partial X_n}\\

\dfrac{\partial g}{\partial X_1} &\cdots & \dfrac{\partial g}{\partial X_n}\\
\end{pmatrix}$$
is equal to one;
\item $D_{f_1/f_2}(g)=0$;
\item there exists $h$ in $\KK(X_1,\ldots,X_n)$ such that $f=u(h)$ and $g=v(h)$ for $u,v \in \KK(T)$.
\end{enumerate}
\end{Prop}

\begin{proof}
See \cite{petravchuk} for a proof. In \cite{petravchuk}, $\KK$ is supposed to be algebraically closed. However, we can remove this hypothesis because we have the equivalence: $f$ is composite over $\KK$ if and only if $f$ is composite over $\overline{\KK}$, see e.g. \cite[Theorem 13]{BuseChezeNajib}.
\end{proof}

%In \cite{Ollagnier}, the author computes the kernel of the linear map $D_{f_1/f_2}$ and then deduce $h$. With this strategy the size of the linear system is $\bigO(d^n)\times \bigO(d^n)$.\\
%The aim of our algorithm is to reduce the number of unknowns in the linear system to $\bigO(d)$.  In order to do this we  introduce the following:

\begin{Def}
Given $D$ a vector derivation i.e. an $m$-tuple of  derivations, a polynomial $F \in \KK[\underline{X}]$ is said to be a Darboux polynomial of $D$ if there exists $\mathcal{G} \in \big(\KK[\underline{X}]\big)^m$ such that $D(F)=F.\mathcal{G}$. $\mathcal{G}$ is called the cofactor of $F$ for the derivation $D$.
\end{Def}

%\begin{Rem}: In this situation, the cofactor $\mathcal{G}$ is unique.\end{Rem}

We deduce easily the following classical propositions.
\begin{Prop}\label{rem_darboux}
$f_1$ and $f_2$ are Darboux polynomials of $D_{f_1/f_2}$.
\end{Prop}

\begin{Prop}\label{cor_same_cof}
$D_{f_1/f_2}(h_1/h_2)=0$ if and only if $h_1$ and $h_2$ are Darboux polynomials with the same cofactor.
\end{Prop}

The following proposition is the main tool of our algorithm. Indeed, this proposition shows that  cofactors transform a product into a sum. Then thanks to the cofactors it will be possible to apply a kind of logarithmic derivative recombination scheme. 

\begin{Prop}\label{prop_cof}
Let $F \in \KK[X_1,\ldots,X_n]$ be a polynomial and let $F=F_1^{e_1}\cdots F_r^{e_r}$ be  its irreducible factorization in  $\KK[X_1,\ldots,X_n]$. Then:\\
$F$ is a Darboux polynomial with cofactor $\mathcal{G}_F$ if and only if all the $F_i$ are Darboux polynomials with cofactor $\mathcal{G}_{F_i}$. Furthermore, $\mathcal{G}_F=e_1 \mathcal{G}_{F_1}+\cdots + e_r\mathcal{G}_{F_r}$.
\end{Prop}

\begin{proof}
See for example Lemma 8.3 page 216 in \cite{Dumortier_Llibre_Artes}.
\end{proof}

%%%%%%%%%%%%%%%%%%%%%%%%%%%%%%%%%%%%%%%
\section{Reduction to a rational factorization problem}\label{sec_reduc}
In this section, we recall how the decomposition problem can be reduced to a factorization problem. Furthermore, we show that we can reduce our problem to a situation where $f_1$ and $f_2$ are squarefree. First, we recall some useful lemmas.

\begin{Lem}\label{Cor_f+gsquarefree}
If $f_1/f_2$ is reduced in $\KK(X_1,\dots,X_n)$, where $n\geq 1$ and $\Lambda$ is a  variable, then $f_1+\Lambda f_2$ is squarefree.
\end{Lem}

\begin{Lem}\label{Lem_Bodin}
Let $h=h_1/h_2$ be a rational function in $\KK(\underline{X})$, $u=u_1/u_2$ a  rational function in $\KK(T)$ and set $f=u \circ h$ with $f=f_1/f_2 \in \KK(\underline{X})$. 
For all $\lambda \in \KK$ such that $\deg (u_1 - \lambda u_2) = \deg u$, we have
$$f_1 - \lambda f_2= e (h_1 - t_1 h_2)\cdots (h_1 -t_k h_2)$$
where $e \in \KK$, $k= \deg u$ and $t_i$ are the roots of the univariate polynomial $u_1(T) - \lambda u_2(T)$. 
\end{Lem}
\begin{proof}
See \cite[Lemma 8, Lemma 39]{Cheze2010}.
\end{proof}

\begin{Rem}\label{rem:K}
 If $\lambda =f_1(\underline{a})/f_2(\underline{a})$, where $\underline{a}=(a_1,\ldots,a_n) \in \KK^n$, then we can suppose that $t_1 \in \KK$. Indeed, $t_1=h_1(\underline{a})/h_2(\underline{a})  \in \KK$.
\end{Rem}

The following lemma says that we can always suppose that $\deg u_1=\deg u_2 =\deg u$.

\begin{Lem}\label{degu1=degu}
Let $h=h_1/h_2$ be a rational function in $\KK(\underline{X})$, $u=u_1/u_2$ a  rational function in $\KK(T)$ and set $f=u \circ h$ with $f=f_1/f_2 \in \KK(\underline{X})$. There exists an homography $H(T)=(aT+b)/(\alpha T + \beta) \in \KK(T)$ such that:\\
$u \circ H= \tilde{u}_1/\tilde{u}_2$, $\deg \tilde{u}_1= \deg \tilde{u}_2$, and $f=\dfrac{\tilde{u}_1}{\tilde{u}_2} \circ \tilde{h}$, where $\tilde{h}=H^{-1} \circ h$ and $H^{-1}$ is the inverse of $H$ for the composition.
\end{Lem}

\begin{proof}
If $\deg u_1 =\deg u_2$ then we set $H(T)=T$.\\
If $\deg u_2 > \deg u_1$ then we have:
$$\dfrac{u_1}{u_2}\big(H(T)\big)=\dfrac{\prod_{i=1}^{\deg u_1} \big( aT+b - \lambda_i(\alpha T + \beta) \big)}{\prod_{i=1}^{\deg u_2} \big( aT+b - \mu_i(\alpha T + \beta) \big)}.(\alpha T + \beta)^{\deg u_2 - \deg u_1},$$
where $u_1(\lambda_i)=0$ and $u_2(\mu_i)=0$.\\
We set:
\begin{eqnarray*}
\tilde{u}_1(T)&=&(\alpha T + \beta)^{\deg u_2 - \deg u_1}.\prod_{i=1}^{\deg u_1} \big( aT+b - \lambda_i(\alpha T + \beta) \big)\\
&=&u_1\big( H(T) \big).(\alpha T + \beta) ^{\deg u_2} \in \KK[T]\\
\tilde{u}_2(T)&=& \prod_{i=1}^{\deg u_2} \big( aT+b - \mu_i(\alpha T + \beta) \big)\\
&=&u_2\big( H(T) \big).(\alpha T + \beta) ^{\deg u_2} \in \KK[T].
\end{eqnarray*}
If $a- \lambda_i \alpha \neq 0$, $\alpha \neq 0$, and $a -\mu_i \alpha \neq 0$ then we get $\deg \tilde{u}_1=\deg u_2=\deg \tilde{u}_2$.\\
To conclude the proof we just have to remark that $\deg H=1$, thus $H$ is invertible for the composition.
\end{proof}

In order to ease the recombination scheme we reduce our problem to a  situation where \emph{the rational function is squarefree, i.e. the numerator and the denominator  are squarefree}. The following algorithm shows that if $f_1$ or $f_2$ are not squarefree then we can compute an homography $U(T) \in \KK(T)$ such that $U(f_1/f_2)$ is squarefree. Furthermore, if we know a decomposition $U(f_1/f_2)=u(h)$ then we can easily deduce  a decomposition $f_1/f_2=U^{-1}\big(u(h)\big)$. We recall that $U$ is invertible for the composition because $\deg U=1$. Now, we describe an algorithm which computes a good homography.\\

\textbf{\textsf{Good homography}}\\
\textbf{Input:} $f=f_1/f_2 \in \KK(X_1,\dots,X_n)$ of degree $d$, such that (C) and (H) are satisfied and a finite subset $S$ of $\KK^n$ such that $|S|= 2d^2+2d$.\\
\textbf{Output:}
 $U(T)=(T-\lambda_a)/(T-\lambda_b)$ such that $U(f)$ is squarefree, $\lambda_a =f_1/f_2(\underline{a}), \lambda_b=f_1/f_2(\underline{b})$ where $\underline{a}, \underline{b} \in \KK^n$, $\lambda_a \neq \lambda_b$, and $\deg_{X_n}( f_1- \lambda_a f_2)=\deg_{X_n}( f_1 - \lambda_b f_2)=d$.\\

\begin{enumerate}
\item \label{comput_fbar}Compute $\overline{f}_1(X_n):=f_1(\underline{0},X_n)$, and $\overline{f}_2(X_n):=f_2(\underline{0},X_n)$.
\item Construct an empty list $L$.
\item \label{construc_L} For $i$ from 1 to $2d^2+2d$ do:
\begin{enumerate}
\item Compute $\overline{f}:=\overline{f}_1(i)/ \overline{f}_2(i)$,.
\item If $\overline{f} \not \in L$ then $L:=\textrm{concatenate}(L,[\overline{f}])$.
\end{enumerate}
\item Construct an empty list $\mathcal{L}$.
\item \label{step_res}For $k$ from 1 to $2d+2$ do:
\begin{enumerate}
\item Compute $R:=Res_{X_n} \Big( \overline{f}_1(X_n)-L[k] \overline{f}_2(X_n), \, \partial_{X_n} \overline{f}_1(X_n) - L[k]  \partial_{X_n} \overline{f}_2(X_n)\Big)$.
\item \label{step5b} If $R\neq 0$ and $\deg _{X_n}(\overline{f}_1 -L[k] \overline{f}_2)=d$, then $\mathcal{L}:=\textrm{concatenate}(\mathcal{L},[L[k]])$.
\end{enumerate}
\item $\lambda_a:=\mathcal{L}[1]$, $\lambda_b:=\mathcal{L}[2]$.
\item Return $U(T)=(T-\lambda_a)/(T-\lambda_b)$.
\end{enumerate}

\begin{Prop}\label{good_hom_correct}
The algorithm \textsf{Good homography} is correct.
\end{Prop}

\begin{proof}
In Step \ref{construc_L} we construct a list with at least  $2d+2$ distinct elements because  $\deg(f)=d$.\\
By hypothesis (H), $R(\Lambda)\neq 0$ and by \cite[Theorem 6.22]{GG}, $\deg(R)\leq 2d-1$. Thus $\mathcal{L}$ contains at least two distincts elements.\\
As $R(\lambda_a)$ and $R(\lambda_b)$ are not equal to zero, and thanks to Step \ref{step5b} the condition on the degree is satisfied,  we deduce that $f_1-\lambda_af_2$ and $f_1 -\lambda_b f_2$ are squarefree.
\end{proof}

\begin{Prop}\label{complex_good_hom}
The algorithm \textsf{Good homography} can be performed  with at most $\bigOsoft(d^{n})$ arithmetic operations over $\KK$.
\end{Prop}

\begin{proof}
Step \ref{comput_fbar} can be done with $\bigOsoft(d^n)$ arithmetic operations with Horner's method.
In  Step \ref{construc_L} we use a fast multipoint  evaluation strategy, then we can perform this step with at most $\bigOsoft(d^2)$ arithmetic operations, see \cite[Corollary 10.8]{GG}. \\
In Step \ref{step_res}, the computation of the resultant can be done with $\bigOsoft(d)$ arithmetic operations, see \cite[Corollary 11.16]{GG}. Thus Step \ref{step_res} can be done with $\bigOsoft(d^{2})$ arithmetic operations.\\
In conclusion the  algorithm can be performed with the desired complexity.
\end{proof}

\begin{Rem}\label{Rem:alphainK}
Suppose $f_1/f_2=v_1/v_2(h)$. With the algorithm \textsf{Good homography} we can write $U(f_1/f_2)=u_1/u_2(h)$ with $u_1/u_2 \in \KK(T)$, $h  \in \KK(X_1,\ldots,X_n)$, and $u_1$ (resp. $u_2$) has a root $\alpha_1$ (resp. $\alpha_2$) in $\KK$. Indeed, we have $u_1=v_1 - \lambda_a v_2$ (resp. $u_2=v_1 - \lambda_b v_2$) and  $\lambda_a =f_1/f_2(\underline{a})$ (resp. $\lambda_b=f_1/f_2(\underline{b})$) then we deduce that $\alpha_1=h_1/h_2(\underline{a})$ (resp. $\alpha_2=h_1/h_2(\underline{b})$).
\end{Rem}

%%%%%%%%%%%%%%%%%%%%%%%%%%%%%%%%%%%%%%%%%%%%%%%%%%%%%%%%%%%

%%%%%%%%%%%%%%%%%%%%%%%%%%%%%%%%%%%%%%%%%%%%%%%%%%
\section{The recombination method}
In this section we describe our recombination method. First, we introduce some notations. By Proposition \ref{rem_darboux}, $F_1$ and $F_2$ are Darboux polynomials of $D_{F_1/F_2}$.We denote by 
$$\mathcal{G}_{F_k}=(\mathcal{G}_{F_k}^{(2)},\ldots,\mathcal{G}_{F_k}^{(n)})$$ 
the cofactor of $F_k$, where $k=1,2$, and \mbox{$\mathcal{G}_{F_k}^{(l)} \in \KK[X_1,\ldots,X_n]$}. We set: 
$$F_k=\prod_{j=1}^{s_k} F_{k,j}$$
 for $k=1,2$, and 
$$\mathcal{G}_{F_{k,j}}=(\mathcal{G}_{F_{k,j}}^{(2)},\ldots,\mathcal{G}_{F_{k,j}}^{(n)}).$$
In $\QQ[\alpha][X_1,\ldots,X_n]$ polynomials are denoted in the following way:
$$\mathcal{P} =\sum_{|\tau| \leq d} \sum_{\epsilon=0}^{r-1} a_{\epsilon,\tau} \alpha^{\epsilon} X_1^{\tau_1}\cdots X_n^{\tau_n} \in \QQ[\alpha][X_1,\ldots,X_n],$$
where $\alpha$ is an algebraic number of degree $r$, $\tau=(\tau_1,\ldots,\tau_n)$, $|\tau|=\tau_1+\cdots+\tau_n$, and $a_{\epsilon,\tau} \in \QQ$. We set
$$\textrm{coef}\big(\mathcal{P},\alpha^{\epsilon}\underline{X}^{\tau}\big)=a_{\epsilon,\tau}.$$
Now we define the linear system $\mathcal{S}$:
$$\mathcal{S}:=\sum_{j=1}^{s_1} x_{1,j}\textrm{coef}\big(\mathcal{G}_{F_{1,j}}^{(l)},\alpha^{\epsilon}\underline{X}^{\tau}\big) - \sum_{j=1}^{s_2} x_{2,j}\textrm{coef}\big(\mathcal{G}_{F_{2,j}}^{(l)},\alpha^{\epsilon}\underline{X}^{\tau}\big)=0,$$
where  $|\tau| \leq d$, $0 \leq \epsilon \leq r-1$, and $2 \leq l\leq n$.\\
We denote by $\ker \mathcal{S}$ the kernel of this linear system, and we remark that 
$$x=(x_{1,1}, \ldots,x_{2,s_2}) \in \ker \mathcal{S} \iff
\sum_{j=1}^{s_1} x_{1,j} \mathcal{G}_{F_{1,j}} - \sum_{j=1}^{s_2} x_{2,j} \mathcal{G}_{F_{2,j}}=0.$$
We define the following maps:

\begin{eqnarray*}
\pi_1: \KK^{s_1+s_2}& \longrightarrow & \KK^{s_1} \\
(x_{1,1},\ldots,x_{2,s_2})& \longmapsto &(x_{1,1}, \ldots,x_{1,s_1})
\end{eqnarray*}
\begin{eqnarray*}
\pi_2: \KK^{s_1+s_2}& \longrightarrow & \KK^{s_2}\\
(x_{1,1},\ldots,x_{2,s_2})& \longmapsto &(x_{2,1}, \ldots,x_{2,s_2})
\end{eqnarray*}

The following proposition will be the key of our algorithm:
\begin{Prop}\label{key_prop}
Suppose that $F_1/F_2 \in \KK(X_1,\ldots,X_n)$ comes from the algorithm \textsf{Good Homography} and $F_1/F_2=u(h)$ where $h=h_1/h_2 \in \KK(X_1,\ldots,X_n)$ is a non-composite reduced rational function and $u=u_1/u_2 \in \KK(T)$ is a reduced rational function, with $\deg u_1= \deg u_2$.\\
We denote by $u_k=\prod_{i=1}^{t_k} u_{k,i}$ the factorization of $u_k$ in $\KK[T]$, where $k=1$, 2.\\
We denote by $F_k=\prod_{j=1}^{s_k}F_{k,j}$ the factorization of $F_k$ in $\KK[X_1,\ldots,X_n]$, where $k=1$, 2.\\
Then:
\begin{enumerate}
\item $\displaystyle u_{k,i}\Big(\dfrac{h_1}{h_2}\Big).h_2^{\deg u_{k,i}}= \prod_{j=1}^{s_k} F_{k,j}^{e_{k,i,j}} \in \KK[X_1,\ldots,X_n]$ and $e_{k,i,j} \in \{0,1\}$.\\
Furthermore, if we set $e_{k,i}:=(e_{k,i,1},\ldots, e_{k,i,s_{k}})$, then\\ the vectors $e_{k,i}$, $i=1,\ldots, t_k$, are orthogonal for the usal scalar product.
\item We have $e_{k,i} \in \pi_k(\ker \mathcal{S})$.
\item $\{e_{k,1}, \ldots, e_{k,t_k}\}$ is a basis of $\pi_k(\ker \mathcal{S})$.
\end{enumerate} 
\end{Prop}

\begin{proof}
\begin{enumerate}
\item By Lemma \ref{Lem_Bodin} applied to $F_1/F_2$ (resp. $F_2/F_1$) with $\lambda=0$,  we get 
$$F_k=u_k(h_1/h_2).h_2^{\deg u_k}= \prod_{i=1}^{t_k} u_{k,i}(h_1/h_2). h_2^{\deg u_{k,i}}.$$
Then we deduce 
$$u_{k,i}(h_1/h_2).h_2^{\deg u_{k,i}} = \prod_{j=1}^{s_k} F_{k,j}^{e_{k,i,j}} \textrm{ in } \KK[X_1, \ldots, X_n]$$ 
with $e_{k,i,j} \in \{0,1\}$ because $F_k$ are squarefree. Furthermore, the vectors  $e_{k,i}$ are orthogonal for the usual scalar product because $F_k$ are squarefree.
\item We show this item for $k=1$, the case $k=2$ can be proved in a similar way.\\
As $F_1/F_2$ comes from the algorithm \textsf{Good Homography} and as explained in Remark \ref{Rem:alphainK} we can suppose that:
$$u_{k,1}(T)=(T-\alpha_k), \textrm{ with } \alpha_k \in \KK.$$
The previous item allows us to write:
$$\dfrac{u_{1,i}}{u_{2,1}^{\deg u_{1,i}}}\Bigg(\dfrac{h_1}{h_2} \Bigg)= 
\dfrac{\Big(\prod_{j=1}^{s_1} F_{1,j}^{e_{1,i,j}}\Big).\Big(h_2\Big)^{\deg u_{1,i}}}{  \Big(\prod_{j=1}^{s_2} F_{2,j}^{e_{2,1,j}}\Big)^{\deg u_{1,i}}.\Big(h_2\Big)^{\deg u_{1,i}}}=
\dfrac{\prod_{j=1}^{s_1} F_{1,j}^{e_{1,i,j}}}{\Big(\prod_{j=1}^{s_2} F_{2,j}^{e_{2,1,j}}\Big)^{\deg u_{1,i}}}.$$
By Proposition \ref{Petrav} applied to $\dfrac{u_{1,i}}{u_{2,1}^{\deg u_{1,i}}}\Big(\dfrac{h_1}{h_2} \Big)$, we get then:
$$D_{F_1/F_2} \Bigg( \dfrac{\prod_{j=1}^{s_1} F_{1,j}^{e_{1,i,j}}}{\prod_{j=1}^{s_2} F_{2,j}^{e_{2,1,j}. \deg u_{1,i}}} \Bigg)=0.$$
Now, we recall that $F_{k,j}$ are Darboux polynomials, see Proposition \ref{rem_darboux} and Proposition \ref{prop_cof}. Then by Proposition \ref{prop_cof}, we deduce 
$$\sum_{j=1}^{s_1} e_{1,i,j} \mathcal{G}_{F_{1,j}}-\deg(u_{1,i}) \sum_{j=1}^{s_2} e_{2,1,j} \mathcal{G}_{F_{2,j}}=0.$$
It follows $(e_{1,i,1},\ldots,e_{1,i,s_1},\deg(u_{1,i}).e_{2,1,1}, \ldots,\deg(u_{1,i}).e_{2,1,s_2}) \in \ker \mathcal{S}$. Thus, $e_{1,i} \in \pi_1(\ker \mathcal{S})$.
\item The vectors  $e_{k,1},\ldots, e_{k,t_k}$ are linearly independant because they are orthogonal. We just have to prove that these vectors generate $\pi_{k} (\ker \mathcal{S})$.\\
Suppose that $\rho=(\rho_1,\ldots,\rho_{s_1+s_2}) \in \ker \mathcal{S}$. First, we clear the denominators and  we  suppose that $\rho \in \ZZ^{s_1+s_2}$ instead of $\QQ^{s_1+s_2}$.\\
In a first time  we explain the strategy of the proof for this item, and in a second time  we will detail the proof.\\
We set $$\dfrac{\mathcal{F}_1}{\mathcal{F}_2}=\dfrac{\prod_{i=1}^{s_1} F_{1,j}^{\rho_j}}{ \prod_{j=1}^{s_2} F_{2,j}^{\rho_{s_1+j}}},$$
 where  $\mathcal{F}_1, \mathcal{F}_2 \in \KK[\underline{X}]$ and $\mathcal{F}_1/\mathcal{F}_2$ is a reduced rational function.\\

 Our goal is to get this kind of equality:
\begin{equation}\nonumber
(\mathcal{E}), \, \,\dfrac{\mathcal{F}_1}{\mathcal{F}_2}=
\dfrac{\prod_{j=1}^{s_1}F_{1,j}^{\rho_j}}{\prod_{j=1}^{s_2}F_{2,j}^{\rho_{s_1+j}}}=
\dfrac{ \prod_{k=1}^2 \prod_{(i,k) \in I_{\textrm{num}}} \Big( \prod_{j=1}^k F_{k,j}^{e_{k,i,j}}   \Big)^{m_{u_{k,i}}}  }{ \prod_{k=1}^2 \prod_{(i,k) \in I_{\textrm{den}}} \Big( \prod_{j=1}^k F_{k,j}^{e_{k,i,j}}   \Big)^{m_{u_{k,i}}}     } ,
%\prod_{i=1}^{t_1}\Big( \prod_{j=1}^{s_1} F_{1,j}^{e_{1,i,j}} \Big)^{m_{u_{1,i}}} \prod_{i=1}^{t_2}\Big(  \prod_{j=1}^{s_2} F_{2,j}^{e_{2,i,j}}     \Big)^{m_{u_{2,i}}},  \quad (A)
\end{equation}
where $m_{u_{k,i}} \in \NN$, $I=\{(1,1),\ldots,(t_1,1),(1,2),\ldots,(t_2,2)\}$, $I_{\textrm{num}} \subset I$, $I_{\textrm{den}} \subset I$ and $I_{\textrm{num}}\cap I_{\textrm{den}}=\emptyset$.\\

By the unicity of the factorization in irreducible factors we deduce:
%$$ v_j=\sum_{i=1}^{t_1} e_{1,i,j} m_{u_1,i}, \textrm{ for } 1\leq j \leq s_1,$$
%$$v_{s_1+j}=-\sum_{i=1}^{t_2} e_{2,i,j} m_{u_2,i}, \textrm{ for } 1\leq j \leq s_2.$$
%$$\rho=\sum_{k=1}^2 \sum_{(i,k) \in I_{\textrm{num}}} m_{u_{k,i}} e_{k,i} - \sum_{k=1}^2\sum_{(i,k) \in I_{\textrm{den}}} m_{u_{k,i}} e_{k,i}.$$
$$\pi_1(\rho)=\sum_{(i,1) \in I_{\textrm{num}}} m_{u_{1,i}} e_{1,i} - \sum_{(i,1) \in I_{\textrm{den}}} m_{u_{1,i}} e_{1,i},$$     
$$\pi_2(\rho)=\sum_{(i,2) \in I_{\textrm{num}}} m_{u_{2,i}} e_{2,i} - \sum_{(i,2) \in I_{\textrm{den}}} m_{u_{2,i}} e_{2,i}.$$

We get: $\{e_{k,1}, \ldots, e_{k,t_{k}} \}$ generates $\pi_k( \ker \mathcal{S})$, and this is the desired result.\\

Now we detail the proof with four steps:

\begin{enumerate}
\item \label{starstar} We remark: $$ \dfrac{F_1}{F_2}=\dfrac{u_1}{u_2}(h)=\dfrac{ \prod_{i=1}^{\deg u} (h_1 - \mu_{1,i} h_2)}{\prod_{i=1}^{\deg u} (h_1 - \mu_{2,i} h_2)},$$
where $\mu_{k,i}$ are roots of $u_{k}$.\\
%Indeed, we apply Lemma \ref{Lem_Bodin} to $F_1/F_2$ (resp. to $F_2/F_1$) with $\lambda=0$.
\item  \label{star} We have:
$$ \dfrac{   \mathcal{F}_1 }{ \mathcal{F}_2 }=\dfrac{ \prod_{j=1}^{d_1}(h_1 -\lambda_j h_2)^{m_j} }{ \prod_{j=d_1+1}^{d_1+d_2}(h_1 -\lambda_j h_2)^{m_j}}.h_2^{\kappa}, \textrm{ with } \kappa \in \ZZ, m_j \in \NN.$$

Indeed, as $\rho \in \ker \mathcal{S}$, we have 
$$\sum_{j=1}^{s_1} \rho_j \mathcal{G}_{F_{1,j}}-\sum_{j=1}^{s_2} \rho_{s_1+j} \mathcal{G}_{F_{2,j}}=0.$$
Thus $\prod_{j=1}^{s_1} F_{1,j}^{\rho_j}$ and $\prod_{j=1}^{s_2}F_{2,j}^{\rho_{s_1+j}}$ are Darboux polynomials with the same cofactor. By Proposition \ref{cor_same_cof}, we deduce:
$$D_{F_1/F_2}\Bigg(\dfrac{\prod_{j=1}^{s_1} F_{1,j}^{\rho_j}}{ \prod_{j=1}^{s_2} F_{2,j}^{\rho_{s_1+j}}}\Bigg)=0.$$
Then $D_{F_1/F_2}(\mathcal{F}_1/\mathcal{F}_2)=0$ and thus $\mathcal{F}_1/\mathcal{F}_2=v_1/v_2(h)$ by Proposition \ref{Petrav}.
We denote by $\lambda_j$ the roots of $v_1$ and $v_2$ and we get the desired result.

\item \label{C} We claim:$$\dfrac{   \mathcal{F}_1 }{ \mathcal{F}_2 }= 
\dfrac{    \prod_{k=1}^2 \prod_{(i,k) \in I_{\textrm{num}}} \Big( u_{k,i}(h_1/h_2)h_2^{\deg u_{k,i}}\Big)^{m_{u_{k,i}}}  }{  \prod_{k=1}^2 \prod_{(i,k) \in I_{\textrm{den}}} \Big( u_{k,i}(h_1/h_2)h_2^{\deg u_{k,i}}\Big)^{m_{u_{k,i}}} }, \textrm{ where } m_{u_k,i} \in \NN.$$

Indeed, we have: for all $j$ there exists $\delta(j)$ such that $\lambda_j=\mu_{\delta(j)}$.\\
(To prove this remark we suppose the converse: There exists $j_0$ such that $\lambda_{j_0} \neq \mu_{k,i}$, for $k=1,2$ and $i=1,\ldots,\deg u$.\\
By definition of $\mathcal{F}_k$ and by step \ref{star}, there exists $(k_1,j_1)$ such that $F_{k_1,j_1}$ and $h_1 - \lambda_{j_0} h_2$ have a common factor in $\CC[\underline{X}]$. We call $\mathcal{P}$ this common factor.\\
By step \ref{starstar}, there exists $(k_2,i_2)$ such that $\mathcal{P}$ is a factor of $h_1 - \mu_{k_2,i_2} h_2$.\\
Thus $h_1-\lambda_{j_0} h_2$ and $h_1 - \mu_{k_2,i_2} h_2$ have a common factor. As $\lambda_{j_0}\neq \mu_{k_2,i_2}$ we deduce that $\mathcal{P}$ divides $h_1$ and $h_2$. This is absurd because $h_1/h_2$ is reduced.)\\
Thus $\kappa=0$, and for all $j$ there exists $k(j) \in \{1,2\}$ and  such that $u_{k(j)}(\lambda_j)=0$. \\
As $v_1, v_2 \in \KK[T]$, by conjugation, we deduce that if $\lambda_j$ and $\lambda_{j'}$ are roots of the same irreducible polynomial $u_{k,i}\in \KK[T]$ then $m_j=m_{j'}$. We denote by $m_{u_{k,i}}$ this common value. \\
This gives the claimed equality with $I_{\textrm{num}}\cap I_{\textrm{den}}=\emptyset$, because $\mathcal{F}_1/\mathcal{F}_2$ is reduced.
\item Now we can prove equality $(\mathcal{E})$.
\begin{eqnarray*}
\dfrac{\mathcal{F}_1}{\mathcal{F}_2}&=&\dfrac{    \prod_{k=1}^2 \prod_{(i,k) \in I_{\textrm{num}}} \Big( u_{k,i}(h_1/h_2)h_2^{\deg u_{k,i}}\Big)^{m_{u_{k,i}}}  }{  \prod_{k=1}^2 \prod_{(i,k) \in I_{\textrm{den}}} \Big( u_{k,i}(h_1/h_2)h_2^{\deg u_{k,i}}\Big)^{m_{u_{k,i}}} }, \textrm{ by step } \ref{C},   \\
 &=&\dfrac{ \prod_{k=1}^2 \prod_{(i,k) \in I_{\textrm{num}}} \Big( \prod_{j=1}^k F_{k,j}^{e_{k,i,j}}   \Big)^{m_{u_{k,i}}}  }{ \prod_{k=1}^2 \prod_{(i,k) \in I_{\textrm{den}}} \Big( \prod_{j=1}^k F_{k,j}^{e_{k,i,j}}   \Big)^{m_{u_{k,i}}}     }, \textrm{ by the first item}.
\end{eqnarray*}
This gives the desired equality $(\mathcal{E})$.

\end{enumerate}
\end{enumerate}

\end{proof}

Now we describe our recombination algorithm:\\

\textbf{\textsf{Recombination for Decomposition}}\\
\textbf{Input:} $f=f_1/f_2 \in \KK(X_1,\dots,X_n)$, such that (C) and (H) are satisfied.\\
\textbf{Output:} A decomposition of $f$ if it exists, with $f=u \circ h$, $u=u_1/u_2$ with $\deg u \geq 2$, and  $h=h_1/h_2$ non-composite.\\

\begin{enumerate}
\item  \label{step1}Compute $F=F_1/F_2:=U(f)$ with the algorithm \textsf{Good homography}.
\item \label{step_facto}For k=1, 2, factorize $F_k=\prod_{i=1}^{s_k}F_{k,i}$  in $\KK[\underline{X}]$ with  $F_{k,i}$ irreducible.
\item \label{step_cof}For each $F_{k,i}$ compute the corresponding cofactor $\mathcal{G}_{F_{k,i}}:=D_{F_1/F_2}(F_{k,i})/F_{k,i}$.
\item \label{step_rref}Build the system $\mathcal{S}$ and compute the basis in reduced row echelon form $\mathcal{B}_1$ of $\pi_1(\ker \mathcal{S})$ and $\mathcal{B}_2$ of $\pi_2(\ker \mathcal{S})$.
\item \label{step_deg} For k=1, 2, find $v_k =(v_{k,1}, \ldots,v_{k,s_k}) \in \mathcal{B}_k$ such that:\\  $\sum_{i=1}^{s_k}v_{k,i}\deg F_{k,i}=\min_{w \in \mathcal{B}_k} \sum_{i=1}^{s_k}w_{i}\deg F_{k,i}$, where $w:=(w_{1},\ldots,w_{s_k})$.
\item \label{step_Hk}For k=1, 2, compute $H_k:=\prod_{i=1}^{s_k} F_{k,i}^{v_{k,i}}$.
\item Set $H:=H_1/H_2$.
\item \label{step_u}Compute $u$ such that $u(H)=f$.
\item Return $H$, and $u$.
\end{enumerate}

\begin{Prop}
The algorithm \textsf{Recombination for Decomposition} is correct.
\end{Prop}

\begin{proof}
Consider $F_1/F_2:=U(f)$. As we want to decompose $f_1/f_2$, we just have to decompose $F_1/F_2$, because $\deg U=1$ and then $U$ is invertible.\\
As $F_1/F_2$ comes from the algorithm \textsf{Good Homography} we can suppose, see Remark \ref{Rem:alphainK}, that  $u_{k,1}(T)=(T-\alpha_k)$ with $\alpha_k \in \KK$, and $k=1, 2$. Furthermore, by Lemma \ref{degu1=degu} we can also suppose that $\deg u_1= \deg u_2$.\\
Then by Proposition \ref{key_prop}, the basis $\mathcal{B}_k$ of $\pi_k(\ker \mathcal{S})$ are $\{e_{k,1}, \ldots, e_{k,t_k} \}$.\\ The vector $e_{k,i}$ gives the polynomial $\mathcal{H}_{k,i}=\prod_{j=1}^{s_k} F_{k,j}^{e_{k,i,j}}=u_{k,i}(h)h_2^{\deg u_{k,i}}$.\\
Furthermore $\deg \mathcal{H}_{k,i}=\sum_{j=1}^{s_k} e_{k,i,j} \deg F_{k,j}= \deg u_{k,i} \deg h$. %because $\deg F_1=\deg F_2=d$.
 Thus in Step \ref{step_deg}
$$\min_{w \in \mathcal{B}_k} \sum_{i=1}^{s_k}w_{i}\deg F_{k,i}=\deg h,$$ because this minimum is reached with $e_{k,1}\in \mathcal{B}_k$. Hence $v_k$ in Step \ref{step_Hk} gives $H_k =u_{k,i(k)}(h)h_2^{\deg u_{k,i(k)}}$ with $\deg u_{k,i(k)}=1$.\\
It follows $H=(h_1 -\alpha h_2)/(h_1 - \beta h_2)$ with $\alpha,\beta \in \KK$. Thus $H=v(h)$ with $\deg v=1$, then the algorithm is correct.
 
\end{proof}

\begin{Prop}
The algorithm \textsf{Recombination for Decomposition} can be performed with $\bigOsoft(rd^{n+\omega-1})$ arithmetic operations over $\QQ$ and two factorizations of univariate polynomials of degree $d$ with coefficients in $\KK$.
\end{Prop}
We recall that in our complexity analysis the number of variables is fixed and the degree $d$ tends to infinity.

\begin{proof}
Step \ref{step1} uses $\bigOsoft(d^n)$ arithmetic operations over $\KK$ by Proposition \ref{complex_good_hom}, thus it uses $\bigOsoft(rd^n)$ arithmetic operations over $\QQ$.\\
Step \ref{step_facto} uses $\bigOsoft(d^{n+\omega-1})$ arithmetic operations over $\KK$ because we can use Lecerf's algorithm, see \cite{Lec2007}. Thus we use $\bigOsoft(rd^{n+\omega-1})$ arithmetic operations over $\QQ$ and two factorizations of univariate polynomials of degree $d$ with coefficients in $\KK$.\\
In Step \ref{step_cof}, we compute $D_{F_1/F_2}(F_{k,i})$, thus we perform $2(n-1)$ multiplications of multivariate polynomials. We can do this with a fast multiplication technique, and then this computation costs $\bigOsoft(nrd^n)$ arithmetic operations over $\QQ$. Then we divide $D_{F_1/F_2}(F_{k,i})$ by $F_{k,i}$. We have to perform $n-1$ exact divisions, thus with a Kronecker subsitution we reduce this problem to $n-1$ univariate divisions, and the cost of one such  division belongs then to $\bigOsoft(rd^{n})$. As $s_1$ and $s_2$ are smaller than $d$, Step \ref{step_cof} costs $\bigOsoft(nrd^{n+1})$ arithmetic operations over $\QQ$.\\
Step \ref{step_rref} needs $\bigOsoft(nrd^nd^{\omega-1})$ arithmetic operations over $\QQ$ with Storjohann's method, see \cite[Theorem 2.10]{Storjohann2000}. Indeed, $\mathcal{S}$ has $\bigO((n-1)rd^{n})$ equations and $s_1+s_2$ unknowns, thus at most $2d$ unknowns.\\
Step \ref{step_deg} has a negligeable cost because $\dim_{\QQ} \pi_k(\ker \mathcal{S})=t_k$ is smaller than $d$ and $s_k$ is also smaller than $d$.\\
In Step \ref{step_Hk}, we use a fast multiplication technique and we compute $H_k$ with $\bigOsoft(rd^{n})$ arithmetic operations over $\QQ$.\\
Step \ref{step_u} can be done with $\bigOsoft(rd^{n})$ arithmetic operations over $\QQ$, see \cite{Cheze2010}.\\
Thus the global cost of the algorithm belongs to $\bigOsoft(rd^{n+\omega-1})$ arithmetic operations over $\QQ$.
\end{proof}

%%%%%%%%%%%%%%%%%%%%%%%%%%%%%%%%%%%%%%%%%%%%%
\section{Examples}\label{sec:examples}
In this section we show the behavior of the algorithm \textsf{Recombination for Decomposition} with two examples. We consider bivariate rational functions with rational coefficients. Thus hypothesis (C) is satisfied.
%%%%%%%%%%%%%%%%%%%%%%%%%%%%%%%
\subsection{$f$ is non-composite}
We set:
\begin{eqnarray*}
f_1&=& \left( 1+X+{Y}^{2} \right)\left( X+Y \right)=X+{X}^{2}+X{Y}^{2}+Y+YX+{Y}^{3},\\
f_2&=&f_1-  \left( {Y}^{2}-X-1 \right) \left( Y-2\,X+1 \right) =-{X}^{2}+3\,X{Y}^{2}+2\,Y+2\,YX-{Y}^{2}+1
\end{eqnarray*}

We have $\deg(f_1 + \Lambda f_2)=\deg_{Y}(f_1 + \Lambda f_2)=3$, and 
$$Res_{Y} \Big( f_1(0,Y)+\Lambda f_2(0,Y), \, \partial_{Y} f_1(0,Y) + \Lambda  \partial_{Y} f_2(0,Y)\Big) = -4-24\,\Lambda-92\,{\Lambda}^{2}-64\,{\Lambda}^{3}+8\,{\Lambda}^{4}.$$
Thus hypothesis (H) is satisfied.\\
The algorithm \textsf{Good homography} gives:\\
 $\lambda_a=f_1(0,0)/f_2(0,0)=0$ and $\lambda_b=f_1(0,1)/f_2(0,1)=1$.\\ 

Then
\begin{eqnarray*}
F_1&=&  \left( 1+X+{Y}^{2} \right) \left( X+Y \right),\\
F_{1,1}&=& 1+X+Y^2,\\
F_{1,2}&=&X+Y\\
F_2&=& \left( {Y}^{2}-X-1 \right) \left( Y-2\,X+1 \right) ,\\
F_{2,1}&=&{Y}^{2}-X-1\\
F_{2,2}&=&Y-2\,X+1\\
\end{eqnarray*}

The cofactors are:
\begin{eqnarray*}
\mathcal{G}_{F_{1,1}}&=& 3\,{X}^{2}+8\,Y{X}^{2}+2\,X-2\,YX+7\,X{Y}^{2}-1+3\,{Y}^{2}-6\,{Y}^{3}-
6\,{Y}^{4}+2\,Y
\\
\mathcal{G}_{F_{1,2}}&=&  3\,{X}^{2}+8\,Y{X}^{2}+4\,YX+6\,X-6\,{Y}^{2}-4\,Y+3-3\,{Y}^{4}-2\,{Y}^
{3}
 \\
\mathcal{G}_{F_{2,1}}&=& 3\,{X}^{2}+8\,Y{X}^{2}+X{Y}^{2}-6\,YX+2\,X-1-2\,Y-6\,{Y}^{4}-11\,{Y}^{
2}-6\,{Y}^{3}
  \\
\mathcal{G}_{F_{2,2}}&=&  3\,{X}^{2}+8\,Y{X}^{2}+6\,X{Y}^{2}+8\,YX+6\,X-3\,{Y}^{4}+8\,{Y}^{2}-2
\,{Y}^{3}+3
 \\
\end{eqnarray*}

The linear system $\mathcal{S}$ is the following:
$$
\left[ \begin {array}{cccc} -1&3&-1&3\\\noalign{\medskip}2&6&2&6\\\noalign{\medskip}3&3&3&3\\\noalign{\medskip}0&0&0&0
\\\noalign{\medskip}0&0&0&0\\\noalign{\medskip}2&-4&-2&0
\\\noalign{\medskip}-2&4&-6&8\\\noalign{\medskip}8&8&8&8
\\\noalign{\medskip}0&0&0&0\\\noalign{\medskip}3&-6&-11&8
\\\noalign{\medskip}7&0&1&6\\\noalign{\medskip}0&0&0&0
\\\noalign{\medskip}-6&-2&-6&-2\\\noalign{\medskip}0&0&0&0
\\\noalign{\medskip}-6&-3&-6&-3\end {array} \right] 
$$

A basis of $\ker(\mathcal{S})$ is given by: $\{(-1,-1,1,1)\}$.\\
Then it follows that $f_1/f_2$ is non-composite.

%%%%%%%%%%%%%%%%%%%%%%%%%%%%%%%%%%%%%%%%%%%%%
\subsection{$f$ is composite} Here we set: 
\begin{eqnarray*}
h_1&=&\left( 1+X+{Y}^{2} \right)\left( X+Y \right)\\
h_2&=&h_1-  \left( {Y}^{2}-X-1 \right) \left( Y-2\,X+1 \right)\\
u_1&=&T.(T-1)\\
u_2&=&T^2+1\\
f_1/f_2&=&u_1/u_2(h_1/h_2).
\end{eqnarray*}
We have constructed a composite rational function $f_1/f_2$ and now we illustrate how our algorithm computes  a decomposition. We can already remark that in the previous example we have shown that $h_1/h_2$ is non-composite. \\

In this situation the hypothesis (H) is satisfied and the algorithm \textsf{Good Homography} gives:\\
 $\lambda_a=f_1(0,0)/f_2(0,0)=0$ and $\lambda_b=f_1(0,2)/f_2(0,2)=90/101$.\\
Then:
\begin{eqnarray*}
F_{1,1}&=&1+X+Y^2\\
F_{1,2}&=&2X-Y-1\\
F_{1,3}&=&Y^2-X-1\\
F_{1,4}&=&X+Y\\
F_{2,1}&=&2\,{X}^{2}+11\,X+9+29\,XY+29\,Y+38\,X{Y}^{2}-9\,{Y}^{2}+11\,{Y}^{3}\\
F_{2,2}&=&11\,{X}^{2}+X-10-19\,YX-19\,Y-29\,X{Y}^{2}+10\,{Y}^{2}+{Y}^{3}
\end{eqnarray*}

The basis in reduced row echelon form of $\pi_1(\ker \mathcal{S})$ (resp. $\pi_2(\ker \mathcal{S})$) is $\{(1,0,0,1); (0,1,1,0)\}$ (resp. $\{(1,0);(0,1)\}$).\\
Step \ref{step_deg} in the algorithm \textsf{Recombination for Decomposition} gives: $v_1=(1,0,0,1)$ and $v_2=(1,0)$.\\
Then we have $H_1:=F_{1,1}.F_{1,4}$ and $H_2:=F_{2,1}$.\\
We remark that $H_1=h_1$ and that $H_2=11h_1+9h_2$. Then $H_1/H_2=w(h_1/h_2)$, where $w(T)=T/(11T+9)$. As $h_1/h_2$ is non-composite and $\deg w=1$, we get a correct output.

%%%%%%%%%%%%%%%%%%%%%%%%%%%%%%%%%%%%%%%%%%%%
\section{Conclusion}\label{sec:conclusion}
In conclusion, we summarize our algorithm with a ``derivation point of view''.\\
In order to decompose $f_1/f_2$, we have computed with Darboux method a rational first integral  of $D_{f_1/f_2}$ with minimum degree. That is to say we have computed $h_1/h_2 \in \KK(X_1,\ldots,X_n)$ such that $D_{f_1/f_2}(h_1/h_2)=0$ and $\deg (h_1/h_2)$ is minimum. In a general setting, Darboux method works as follows: If we want to compute a rational first integral of a derivation $D$, first we compute all the Darboux polynomials $F_i$ and their associated cofactors $\mathcal{G}_{F_i}$, second we solve the linear system
$$\sum_i e_i \mathcal{G}_{F_i}=0.$$
Then thanks to Proposition \ref{prop_cof}, we deduce that $\prod_i F_i^{e_i}$ is a first integral, i.e. \mbox{$D(\prod_i F_i^{e_i})=0$}.\\

When we consider the derivation $D_{f_1/f_2}$ the computation of Darboux polynomials is reduced to the  factorization of $f_1 +\lambda f_2$. Thus this step can be done efficiently. In the general setting, we can also reduce the computation of Darboux polynomials to a factorization problem, see  \cite{ChezeDarboux}.\\
During the second step, we compute the kernel of $\sum_i e_i \mathcal{G}_{F_i}=0$. It is actually a recombination step. Indeed, this system explains how we have to recombine $F_i$ in order to get a rational first integral. Furthermore, the cofactor $\mathcal{G}_{F_i}=D(F_i)/F_i$ can be viewed as a logarithmic derivative.\\
In conclusion, the recombination scheme used in this paper is called nowadays the logarithmic derivative method, but this method is Darboux original method.

%%%%%%%%%%%%%%%%%%%%%%%%%%%%%%%%%%%%%%%%%%%%%%%%%%%%
\appendix
\section{Convex-dense bivariate decomposition}
In this appendix we give complexity results for  the decomposition of sparse bivariate rational functions. These results rely on a strategy proposed by J.~Berthomieu and G.~Lecerf in \cite{BL2010}.\\

Given a polynomial $f(X,Y) \in \KK[X,Y]$, its support is the set $S_f$ of integer points $(i;j)$ such that the monomial $X^iY^j$ appears in $f$ with a non zero coefficient. The convex hull, in the real space $\RR^2$ of $S_f$ is denoted by $N(f)$ and called the Newton's polygon of $f$. We denote by $|N(f)|$ the number of integral points of $N(f)$. We called $|N(f)|$ the convex-size of $f$.\\

Roughly speaking, the transformation proposed in \cite{BL2010} consists in a monomial transformation that preserves the convex-size  but decreases the dense size. The considered transformation $\mathcal{T}$ can be described in the following way:
$$\mathcal{T}=\mathcal{B}\circ \mathcal{L}, \textrm{ where}$$
$$\mathcal{B}(X^iY^j)=X^{i+b_1}Y^{j+b_2}, \, b_1, b_2 \in \ZZ,$$
$$\mathcal{L}(X^iY^j)=X^{a_1 i + a_2 j}Y^{a_3 i+a_4 j}, \, a_1a_4-a_2a_3=\pm1.$$

$\mathcal{T}$ can be defined on $\KK[X,Y,X^{-1},Y^{-1}]$, and we define: $\mathcal{T}(\sum_{i,j} f_{i,j}X^iY^j)=\sum_{i,j} f_{i,j}\mathcal{T}(X^iY^j)$.\\
The transformation $\mathcal{L}$ corresponds to the linear map: $(i,j) \mapsto \mathcal{A} \,^t(i,j)$, where 
$$\mathcal{A}=
\begin{pmatrix}
a_1&a_2\\
a_3&a_4
\end{pmatrix}.$$
We denote by $\mathcal{L}^{-1}$ the transformation corresponding to $\mathcal{A}^{-1}$.\\
If $f(X,Y) \in \KK[X,Y]$, then $\mathcal{L}(f) \in \KK[X,Y,X^{-1},Y^{-1}]$ and $\mathcal{L}(f)$ can be written $\mathcal{L}(f)=c_{\mathcal{L}}(f).\mathcal{L}_0(f)$, where $\mathcal{L}_0(f) \in  \KK[X,Y]$ and $c_{\mathcal{L}}(f)=X^iY^j \in \KK[X,Y,X^{-1},Y^{-1}]$. Furthermore, we also have $\mathcal{L}(F_1.F_2)=\mathcal{L}(F_1).\mathcal{L}(F_2)$.

Let $S$ be a finite subset of $\ZZ^2$. Set $S$ is said to be normalized if  it belongs to $\NN^2$ and if it contains at least one point in $\{0\}\times \NN$, and also at least one point in $\NN \times \{0\}$. For such a normalized set, we write $d_x$ (resp. $d_y$) for the largest abscissa (resp. ordinate) involved in $S$, so that the bounding rectangle is $\mathcal{R}=[0,d_x]\times [0,d_y]$. The following result is proved in \cite[Theorem 2]{BL2010}:\\
\emph{For any normalized finite subset $S$ of $\ZZ^2$, of cardinality $\sigma$, convex-size $\pi$, and bounding rectangle $[0,d_x]\times[0,d_y]$, and dense size $\delta=(d_x+1)(d_y+1)$, one can compute an affine map $\mathcal{T}=\mathcal{B}\circ \mathcal{L}$, with $\bigO(\sigma \log^2\delta)$ bit-operations, such that $\mathcal{T}(S)$ is normalized of dense size at most $9\pi$.}

We are going to use this transformation in order to prove:

\begin{Thm}\label{thm:convexcomplexity}
Let $f_1/f_2(X,Y) \in \KK(X,Y)$ such that $\deg(f_1/f_2)=d$, $N(f_1)\subset N$, $N(f_2) \subset N$ and $N$ is normalized. Then 
\begin{enumerate}
\item If $\KK$ is field with characteristic $0$ or at least $d(d-1)+1$ and (H) is satisfied, then there exists a probabilistic algorithm which computes the decomposition of $f_1/f_2$ with at most $\bigOsoft(|N|^{1,5})$ operations in $\KK$ and two factorizations of a univariate polynomial of degree at most $9 |N|$ over $\KK$.
\item If (C) and (H) are satisfied, then there exists a deterministic algorithm which computes the decomposition of $f_1/f_2$ with at most $\bigOsoft(r.|N|^{(\omega+1)/2})$ operations over $\QQ$ and two factorizations of an univariate polynomials of degree at most $9|N|$ over $\QQ[\alpha]$.
\end{enumerate}
\end{Thm}

Now, we explain how we use the transformation $\mathcal{T}$ in the decomposition setting.
\begin{Prop}\label{prop:T(f)}
If $f_1/f_2=u(h_1/h_2)$ then $\mathcal{T}(f_1)/\mathcal{T}(f_2)=u\big(\mathcal{L}(h_1)/\mathcal{L}(h_2)\big)$.\\
If $\mathcal{T}(f_1)/\mathcal{T}(f_2)=u(H_1/H_2)$ then  $f_1/f_2=u\big(\mathcal{L}^{-1}(H_1)/\mathcal{L}^{-1}(H_2)\big)$.
\end{Prop}
\begin{proof}
We prove the first item, the second can be proved in a similar way.\\
We have: $\dfrac{f_1}{f_2}=\dfrac{\prod_i (h_1- \mu_{1,i} h_2)}{\prod_{j} (h_1 - \mu_{2,j} h_2)}$, where $\mu_{k,i}$ are roots of $u_k$. Then,
\begin{eqnarray*}
\dfrac{\mathcal{T}(f_1)}{\mathcal{T}(f_2)}&=& \dfrac{\mathcal{B}\circ \mathcal{L} (f_1)}{\mathcal{B}\circ \mathcal{L} (f_2)}=\dfrac{X^{b_1}Y^{b_2} \mathcal{L}(f_1)}{X^{b_1}Y^{b_2} \mathcal{L}(f_2)}\\
&=&\dfrac{\mathcal{L}(f_1)}{ \mathcal{L}(f_2)}=\dfrac{\prod_i \big( \mathcal{L}(h_1) - \mu_{1,i} \mathcal{L}(h_2)\big)}{\prod_j \big(\mathcal{L}(h_1) - \mu_{2,j} \mathcal{L}(h_2)\big)}\\
&=&\dfrac{ u_1\big( \mathcal{L}(h_1)/\mathcal{L}(h_2) \big)}{u_2\big( \mathcal{L}(h_1)/\mathcal{L}(h_2) \big)}=u\big(\mathcal{L}(h_1)/\mathcal{L}(h_2)\big)
\end{eqnarray*}
\end{proof}

This gives the following algorithm:\\

\textbf{\textsf{Convex bivariate decomposition}}\\
\textbf{Input:} $f=f_1/f_2 \in \KK(X,Y)$, where $N(f_1) \subset N$, $N(f_2) \subset N$ and $N$ is normalized.\\
\textbf{Output:} A decomposition of $f$ if it exists, with $f=u \circ h$, $u=u_1/u_2$ with $\deg u \geq 2$, and  $h=h_1/h_2$ non-composite.\\

\begin{enumerate}
\item Compute $F=\mathcal{T}(f_1)/\mathcal{T}(f_2)$.
\item Decompose $F=u(H)$.
\item Return $f=u(h)$, where $h=\dfrac{\mathcal{L}^{-1}(H_1)}{\mathcal{L}^{-1}(H_2)}=\dfrac{c_{\mathcal{L}^{-1}}(H_1) .\mathcal{L}_0^{-1}(H_1)}{ c_{\mathcal{L}^{-1}}(H_2) .\mathcal{L}_0^{-1}(H_2)} \in \KK(X,Y)$.
\end{enumerate}

\begin{Prop}\label{prop:convexcorrect}
The algorithm \textsf{Convex bivariate decomposition} is correct.
\end{Prop}
\begin{proof}
This follows from Proposition \ref{prop:T(f)}.
\end{proof}

\begin{Prop}\label{prop:convexcomp}
The algorithm \textsf{Convex bivariate decomposition} uses one decomposition of a rational function of degree at most $9|N|$ and $\bigO(\sigma^2\delta)$ bit operations.
\end{Prop}
\begin{proof}
We apply \cite[Theorem 2]{BL2010} to $N$.
\end{proof}

The proof of Theorem \ref{thm:convexcomplexity} comes from Proposition \ref{prop:convexcorrect} and Proposition \ref{prop:convexcomp} and  complexity results given in Theorem \ref{main_thm} and \cite[Theorem 2]{Cheze2010}.

%%%%%%%%%%%%%%%%%%%%%%%%%%%%%%%%%%%%%%%%%%%%%%%%%%%
\newcommand{\etalchar}[1]{$^{#1}$}

%%%%%%%%%%%%%%%%%%
%\bibliographystyle{alpha}
%\bibliography{biblio_decomp_darboux}

\begin{thebibliography}{BHKS09}

\bibitem[AGR95]{Alonso_Gut_Recio}
C.~Alonso, J.~Gutierrez, and T.~Recio.
\newblock A rational function decomposition algorithm by near-separated
  polynomials.
\newblock {\em J. Symbolic Comput.}, 19(6):527--544, 1995.

\bibitem[AT85]{Alagar_Thanh}
V.~S. Alagar and Mai Thanh.
\newblock Fast polynomial decomposition algorithms.
\newblock In {\em E{UROCAL} '85, {V}ol.\ 2 ({L}inz, 1985)}, volume 204 of {\em
  Lecture Notes in Comput. Sci.}, pages 150--153. Springer, Berlin, 1985.

\bibitem[BCN]{BuseChezeNajib}
L.~Bus{\'e}, G.~Ch{\`e}ze, and S.~Najib.
\newblock Noether's forms for the study of non-composite rational functions and
  their spectrum.
\newblock {\em A. Arithmetica}, to appear.

\bibitem[BDN09]{BodinDebesNajib}
A.~Bodin, P.~D{\`e}bes, and S.~Najib.
\newblock Indecomposable polynomials and their spectrum.
\newblock {\em A. Arithmetica}, 139(1):79--100, 2009.

\bibitem[BHKS09]{Bellabas}
K.~Belabas, {M. van} Hoeij, J.~Kl\"uners, and A.~Steel.
\newblock Factoring polynomials over global fields.
\newblock {\em J. Theorie des Nombres de Bordeaux}, 21:15--39, 2009.

\bibitem[BL10]{BL2010}
J.~Berthomieu and G.~Lecerf.
\newblock Convex-dense bivariate polynomial factorization.
\newblock Manuscript, 2010.

\bibitem[BLS{\etalchar{+}}04]{Bostan}
A.~Bostan, G.~Lecerf, B.~Salvy, \'E. Schost, and B.~Wiebelt.
\newblock Complexity {I}ssues in {B}ivariate {P}olynomial {F}actorization.
\newblock In {\em Proceedings of ISSAC 2004}, pages 42--49. ACM, 2004.

\bibitem[BP94]{BiPa1994}
D.~Bini and V.Y. Pan.
\newblock {\em Polynomial and matrix computations. {V}ol. 1}.
\newblock Progress in Theoretical Computer Science. Birkh\"auser Boston Inc.,
  Boston, MA, 1994.
\newblock Fundamental algorithms.

\bibitem[BZ85]{Barton_Zippel}
D.R. Barton and R.~Zippel.
\newblock Polynomial decomposition algorithms.
\newblock {\em J. Symbolic Comput.}, 1(2):159--168, 1985.

\bibitem[Ch{\`e}10]{Cheze2010}
G.~Ch{\`e}ze.
\newblock Nearly optimal algorithms for the decomposition of multivariate
  rational functions and the extended {L}uroth's theorem.
\newblock {\em J. Complexity}, 26(4):344--363, 2010.

\bibitem[Ch{\`e}11]{ChezeDarboux}
G.~Ch{\`e}ze.
\newblock Computation of {D}arboux polynomials and rational first integrals
  with bounded degree in polynomial time.
\newblock {\em J. Complexity}, to appear 2011.

\bibitem[CL07]{CL}
G.~Ch{\`e}ze and G.~Lecerf.
\newblock Lifting and recombination techniques for absolute factorization.
\newblock {\em J. Complexity}, 23(3):380--420, 2007.

\bibitem[CN10]{ChezeNajib}
G.~Ch{\`e}ze and S.~Najib.
\newblock Indecomposability of polynomials via jacobian matrix.
\newblock {\em J. Algebra}, 324(1):1--11, 2010.

\bibitem[Dic87]{Dickerson}
M.~Dickerson.
\newblock Polynomial decomposition algorithms for multivariate polynomials.
\newblock Technical Report TR87-826, Comput. Sci., Cornell Univ., 1987.

\bibitem[DLA06]{Dumortier_Llibre_Artes}
F.~Dumortier, J.~Llibre, and J.C. Art{\'e}s.
\newblock {\em Qualitative theory of planar differential systems}.
\newblock Universitext. Springer-Verlag, Berlin, 2006.

\bibitem[FGP10]{faugere_vonzur}
J.-C. Faug{\`e}re, J.~{von zur} Gathen, and L.~Perret.
\newblock Decomposition of generic multivariate polynomials.
\newblock In {\em Proceedings of ISSAC 2010}, pages 131--137. ACM, 2010.

\bibitem[FP09]{faugere_perret}
J.-C. Faug{\`e}re and L.~Perret.
\newblock An efficient algorithm for decomposing multivariate polynomials and
  its applications to cryptography.
\newblock {\em J. Symbolic Comput.}, 44(12):1676--1689, 2009.

\bibitem[Gat90a]{Gathen_tame}
J.~{von zur} Gathen.
\newblock Functional decomposition of polynomials: the tame case.
\newblock {\em J. Symbolic Comput.}, 9(3):281--299, 1990.

\bibitem[Gat90b]{Gathen_wild}
J.~{von zur} Gathen.
\newblock Functional decomposition of polynomials: the wild case.
\newblock {\em J. Symbolic Comput.}, 10(5):437--452, 1990.

\bibitem[Gat08]{vonzur_denom}
J.~{von zur} Gathen.
\newblock Counting decomposable multivariate polynomials.
\newblock arXiv:0811.4726v2, 2008.

\bibitem[GG03]{GG}
J.~{von zur} Gathen and J.~Gerhard.
\newblock {\em Modern computer algebra}.
\newblock Cambridge University Press, Cambridge, second edition, 2003.

\bibitem[GGR03]{GatGutRub}
J.~{von zur} Gathen, J.~Gutierrez, and R.~Rubio.
\newblock Multivariate polynomial decomposition.
\newblock {\em Appl. Algebra Engrg. Comm. Comput.}, 14(1):11--31, 2003.

\bibitem[Gie88]{Giesbrecht}
M.~Giesbrecht.
\newblock Some results on the functional decomposition of polynomials.
\newblock Master Thesis, University of Toronto, arXiv:1004.5433, 1988.

\bibitem[GRS01]{Gut_Rub_Sev}
J.~Gutierrez, R.~Rubio, and D.~Sevilla.
\newblock Unirational fields of transcendence degree one and functional
  decomposition.
\newblock In {\em ISSAC '01: Proceedings of the 2001 international symposium on
  Symbolic and algebraic computation}, pages 167--174, New York, NY, USA, 2001.
  ACM Press.

\bibitem[GW95]{Gathen_hom_biv}
J.~{von zur} Gathen and J.~Weiss.
\newblock Homogeneous bivariate decompositions.
\newblock {\em J. Symbolic Comput.}, 19(5):409--434, 1995.

\bibitem[KL89]{Kozen_Landau}
D.~Kozen and S.~Landau.
\newblock Polynomial decomposition algorithms.
\newblock {\em J. Symbolic Comput.}, 7(5):445--456, 1989.

\bibitem[Kl{\"u}99]{Kluners}
J.~Kl{\"u}ners.
\newblock On polynomial decompositions.
\newblock {\em J. Symbolic Comput.}, 27(3):261--269, 1999.

\bibitem[Lec06]{Lec2006}
G.~Lecerf.
\newblock Sharp precision in {H}ensel lifting for bivariate polynomial
  factorization.
\newblock {\em Math. Comp.}, 75(254):921--933 (electronic), 2006.

\bibitem[Lec07]{Lec2007}
G.~Lecerf.
\newblock Improved dense multivariate polynomial factorization algorithms.
\newblock {\em J. Symbolic Comput.}, 42(4):477--494, 2007.

\bibitem[MO04]{Ollagnier}
J.~Moulin~Ollagnier.
\newblock Algebraic closure of a rational function.
\newblock {\em Qual. Theory Dyn. Syst.}, 5(2):285--300, 2004.

\bibitem[PI07]{petravchuk}
A.P. Petravchuk and O.G. Iena.
\newblock On closed rational functions in several variables.
\newblock {\em Algebra Discrete Math.}, (2):115--124, 2007.

\bibitem[Rit22]{Ritt}
J.F. Ritt.
\newblock Prime and composite polynomials.
\newblock {\em Trans. Amer. Math. Soc.}, 23(1):51--66, 1922.

\bibitem[Sto00]{Storjohann2000}
A.~Storjohann.
\newblock {\em Algorithms for matrix canonical forms}.
\newblock PhD thesis, ETH Zurich, Zurich, Switzerland, 2000.

\bibitem[Wat08]{watt_laurent}
S.~Watt.
\newblock Functional decomposition of symbolic polynomials.
\newblock In {\em International Conference on Computatioanl Sciences and its
  Applications}, pages 353--362. IEEE Computer Society, 2008.

\bibitem[Wat09]{watt}
S.~Watt.
\newblock Algorithms for the functional decomposition of laurent polynomials.
\newblock In {\em Conferences on Intelligent Computer Mathematics 2009: 16th
  Symposium on the Integration of Symbolic Computation and Mechanized Reasoning
  and 8th International Conference on Mathematical Knowledge Management ,
  (Calculemus 2009)}, pages 186--200. Springer-Verlag LNAI 5625, 2009.

\bibitem[Wei10]{Weiman}
M.~Weimann.
\newblock A lifting and recombination algorithm for rational factorization of
  sparse polynomials.
\newblock {\em J. Complexity}, 26(6):608--628, 2010.

\bibitem[Zip91]{Zippel}
R.~Zippel.
\newblock Rational function decomposition.
\newblock In {\em Proceedings of the 1991 international symposium on Symbolic
  and algebraic computation}, pages 1--6. ACM Press, 1991.

\end{thebibliography}

%%%%%%%%%%%%%%%%%%%%%%%

\end{document}